\documentclass[11pt]{article}

\pdfoutput=1

\usepackage{graphicx, mathtools, amssymb, latexsym, amsmath, amsfonts, amsthm, bbm, tikz}
\usepackage[toc,page]{appendix}
\usepackage[ruled, linesnumbered, noend]{algorithm2e}
\usepackage{enumitem}

\usepackage{array, fancyhdr, bm, listings, soul, xcolor,titlesec, cancel}
\usepackage[skip=0.5\baselineskip]{caption}

\definecolor{Darkblue}{rgb}{0,0,0.4}
\definecolor{Brown}{cmyk}{0,0.61,1.,0.60}
\definecolor{Purple}{cmyk}{0.45,0.86,0,0}
\definecolor{Darkgreen}{rgb}{0.133,0.543,0.133}
\usepackage[pdftex, plainpages = false, pdfpagelabels, 
                 bookmarks=false,
                 bookmarksopen = true,
                 bookmarksnumbered = true,
                 breaklinks = true,
                 linktocpage,
                 pagebackref,
                 colorlinks = true,  
                 linkcolor = blue,
                 urlcolor  = blue,
                 citecolor = blue,
                 anchorcolor = green,
                 hyperindex = true,
                 hypertexnames = false,
                 hyperfigures
                 ]{hyperref}

\usepackage[nameinlink]{cleveref}
\crefname{observation}{Observation}{Observations}
\crefname{claim}{Claim}{Claims}

\theoremstyle{plain}
\newtheorem{theorem}{Theorem}[section]
\newtheorem{lemma}[theorem]{Lemma}

\newtheorem{observation}[theorem]{Observation}
\newtheorem{proposition}[theorem]{Proposition}

\newtheorem*{lemma*}{Lemma}
\newtheorem*{claim*}{Claim}
\newtheorem*{proposition*}{Proposition}
\newtheorem*{fact*}{Fact}
\newtheorem*{corollary*}{Corollary}
\newtheorem*{hint*}{Hint}

\theoremstyle{definition}
\newtheorem{definition}[theorem]{Definition}

\newtheorem*{theorem*}{Theorem}
\newtheorem*{definition*}{Definition}
\newtheorem*{remark*}{Remark}
\newtheorem*{notation*}{Notation}
\newtheorem*{example*}{Example}
\newtheorem*{examples*}{Examples}
\newtheorem*{question*}{Question}
\newtheorem*{answer*}{Answer}
\newtheorem*{problem*}{Problem}
\newtheorem*{solution*}{Solution}
\newtheorem*{idea*}{Idea}
\newtheorem*{conjecture*}{Conjecture}

\newcommand{\OO}{\mathcal{O}}

\mathtoolsset{showonlyrefs}

\begin{document}
\hypersetup{pageanchor=false}

\title{A Single-Exponential Time 2-Approximation Algorithm for Treewidth}
\author{Tuukka Korhonen\thanks{Department of Informatics, University of Bergen, Norway. (\texttt{ tuukka.korhonen@uib.no}). Work done while at University of Helsinki, Finland.}}
\date{}
\maketitle
\thispagestyle{empty}

\begin{abstract}
We give an algorithm that, given an $n$-vertex graph $G$ and an integer $k$, in time $2^{\OO(k)} n$ either outputs a tree decomposition of $G$ of width at most $2k + 1$ or determines that the treewidth of $G$ is larger than $k$.
This is the first 2-approximation algorithm for treewidth that is faster than the known exact algorithms, and in particular improves upon the previous best approximation ratio of 5 in time $2^{\OO(k)} n$ given by Bodlaender et al. [SIAM J. Comput., 45 (2016)].
Our algorithm works by applying incremental improvement operations to a tree decomposition, using an approach inspired by a proof of Bellenbaum and Diestel [Comb.~Probab. Comput.,~11~(2002)].
\end{abstract}

\newpage
\hypersetup{pageanchor=true}
\pagestyle{plain}
\pagenumbering{arabic}

\section{Introduction}
Treewidth is a fundamental graph parameter, playing a central role in multiple fields.
In particular, many graph problems that are intractable in general can be solved in time $f(k) n$ when the input includes also a tree decomposition of the graph of width $k$~\cite{DBLP:journals/iandc/Courcelle90}.
For a large number of classical NP-hard graph problems there are in fact such algorithms with time complexity $2^{\OO(k)} n$~\cite{DBLP:journals/dam/ArnborgP89,DBLP:conf/icalp/Bodlaender88,DBLP:journals/iandc/BodlaenderCKN15}.
To use these algorithms, it is crucial to also have an algorithm for finding a tree decomposition with near-optimal width.
In particular, in order to truly obtain algorithms with time complexity $2^{\OO(k)} n$ for these problems, where $k$ is the treewidth of the input graph, one needs to be able to compute a tree decomposition of width at most $ck$, for some constant $c$, in time $2^{\OO(k)} n$.
Moreover, lowering the constant $c$ directly speeds up all of these algorithms.

There is a long history of algorithms for finding tree decompositions with different guarantees on the width of the decomposition and on the time complexity of the algorithm.
See \Cref{tab:history} for an overview of the most relevant of these results.
The first constant-factor approximation algorithm with time complexity of type $f(k) n^{\OO(1)}$ was given by Robertson and Seymour in their graph minors series~\cite{DBLP:journals/jct/RobertsonS95b}.
The dependency on $n$ in this algorithm was improved to $n \log n$ by Reed~\cite{DBLP:conf/stoc/Reed92}, at the cost of a worse approximation ratio and dependency on $k$.
Bodlaender introduced the first algorithm for treewidth with a linear dependency on $n$~\cite{DBLP:journals/siamcomp/Bodlaender96}.
The algorithm of Bodlaender in fact computes a tree decomposition of optimal width, but with a running time dependency of $2^{\OO(k^3)}$ on the width $k$.

\begin{table*}[!b]
\centering
\caption{Overview of treewidth algorithms that, given an $n$-vertex graph $G$ and an integer $k$, in time $f(k) \cdot g(n)$ either output a tree decomposition of width at most $\alpha(k)$ or determine that the treewidth of $G$ is larger than $k$.}
\label{tab:history}
\begin{tabular}{|c | c | c | c|}
\hline
Reference & $\alpha(k)$ & $f(k)$ & $g(n)$\\
\hline
Arnborg, Corneil, and Proskurowski~\cite{arnborg1987complexity} & $k$ & $\OO(1)$ & $n^{k+2}$\\
Robertson and Seymour~\cite{DBLP:journals/jct/RobertsonS95b} & $4k + 3$ & $\OO(3^{3k})$ & $n^2$\\
Lagergren~\cite{DBLP:journals/jal/Lagergren96} & $8k+7$ & $2^{\OO(k \log k)}$ & $n \log^2 n$\\
Reed~\cite{DBLP:conf/stoc/Reed92} & $8k + \OO(1)$ & $2^{\OO(k \log k)}$ & $n \log n$\\
Bodlaender~\cite{DBLP:journals/siamcomp/Bodlaender96} & $k$ & $2^{\OO(k^3)}$ & n\\
Amir~\cite{DBLP:journals/algorithmica/Amir10} & $4.5k$ & $\OO(2^{3k} k^{3/2})$ & $n^2$\\
Amir~\cite{DBLP:journals/algorithmica/Amir10} & $(3 + 2/3)k$ & $\OO(2^{3.7k})$ & $n^2$\\
Amir~\cite{DBLP:journals/algorithmica/Amir10} & $\OO(k \log k)$ & $\OO(k \log k)$ & $n^4$\\
Feige, Hajiaghayi, and Lee~\cite{DBLP:journals/siamcomp/FeigeHL08} & $\OO(k \sqrt{\log k})$ & $\OO(1)$ & $n^{\OO(1)}$\\
Fomin, Todinca, and Villanger~\cite{DBLP:journals/siamcomp/FominTV15} & $k$ & $\OO(1)$ & $1.7347^n$\\
Fomin et al.~\cite{DBLP:journals/talg/FominLSPW18} & $\OO(k^2)$ & $\OO(k^7)$ & $n \log n$\\
Belbasi and F{\"{u}}rer~\cite{DBLP:conf/cocoa/BelbasiF21} & $5k+4$ & $\OO(2^{6.76k})$ & $n \log n$\\
Bodlaender et al.~\cite{DBLP:journals/siamcomp/BodlaenderDDFLP16} & $3k +4$ & $2^{\OO(k)}$ & $n \log n$\\
Bodlaender et al.~\cite{DBLP:journals/siamcomp/BodlaenderDDFLP16} & $5k+4$ & $2^{\OO(k)}$ & $n$\\
This paper & $2k+1$ & $2^{\OO(k)}$ & $n$\\
\hline
\end{tabular}
\end{table*}

While after the first half of the 1990s multiple improvements to treewidth approximation were given~\cite{DBLP:journals/algorithmica/Amir10,DBLP:journals/siamcomp/FeigeHL08}, the problem of constant-factor approximating treewidth in time $2^{\OO(k)} n$ stood until 2013, when Bodlaender, Drange, Dregi, Fomin, Lokshtanov, and Pilipczuk gave a $2^{\OO(k)} n$ time 5-approximation algorithm for treewidth~\cite{DBLP:journals/siamcomp/BodlaenderDDFLP16}.
In the same article they also gave a 3-approximation algorithm with time complexity $2^{\OO(k)} n \log n$.
Prior to the present work, the aforementioned 5-approximation algorithm was the only constant-factor approximation algorithm with time complexity $2^{\OO(k)} n$ and the 3-approximation algorithm had the best approximation ratio achieved in time $2^{\OO(k)} n^{\OO(1)}$.

In this paper we improve upon both of the algorithms given in~\cite{DBLP:journals/siamcomp/BodlaenderDDFLP16}.

\begin{theorem}
\label{the:main}
There is an algorithm that, given an $n$-vertex graph $G$ and an integer $k$, in time $2^{\OO(k)} n$ either outputs a tree decomposition of $G$ of width at most $2k + 1$ or determines that the treewidth of $G$ is larger than $k$.
\end{theorem}

To further compare our algorithm to the results of Bodlaender et al., we remark that our algorithm has significantly smaller exponential dependency on $k$ hidden in the $2^{\OO(k)}$ factor than what the techniques of~\cite{DBLP:journals/siamcomp/BodlaenderDDFLP16} yield, although we note that the main goal of neither our work nor their work was to optimize this factor.
Our algorithm also makes progress in that it is the first treewidth approximation algorithm to significantly deviate from the basic shape introduced by Robertson and Seymour~\cite{DBLP:journals/jct/RobertsonS95b}.
In particular, all previous treewidth approximation algorithms~\cite{DBLP:journals/algorithmica/Amir10,DBLP:conf/cocoa/BelbasiF21,DBLP:journals/siamcomp/BodlaenderDDFLP16,DBLP:journals/jal/BodlaenderGHK95,DBLP:journals/siamcomp/FeigeHL08,DBLP:journals/talg/FominLSPW18,DBLP:journals/jal/Lagergren96,DBLP:conf/stoc/Reed92,DBLP:journals/jct/RobertsonS95b} follow the same basic structure of building the tree decomposition in a top-down manner from the root to the leafs by using small balanced separators.


Our algorithm is instead based on iteratively making local improvements to a tree decomposition, with the method of improvement inspired by the work of Bellenbaum and Diestel~\cite{DBLP:journals/cpc/BellenbaumD02}.
To the best of the author's knowledge, our algorithm is the first to apply the technique of Bellenbaum and Diestel in the context of computing treewidth.
The technique of Bellenbaum and Diestel has been applied before~\cite{DBLP:journals/talg/CyganKLPPSW21,DBLP:conf/focs/Lokshtanov0S20} for optimizing a different criterion on tree decompositions, with applications to improved parameterized algorithms for minimum bisection, Steiner cut, and Steiner multicut~\cite{DBLP:journals/talg/CyganKLPPSW21}, and for obtaining a parameterized approximation scheme for minimum $k$-cut~\cite{DBLP:conf/focs/Lokshtanov0S20}.

Before describing our algorithm we further mention some related work.
There are multiple refinements of the $2^{\OO(k^3)}n$ time exact treewidth algorithm of Bodlaender~\cite{DBLP:journals/siamcomp/Bodlaender96}, including in particular the version of Perkovi\'{c} and Reed~\cite{DBLP:journals/ijfcs/PerkovicR00} that has applications to the disjoint paths problem and the logarithmic space version of Elberfeld, Jakoby, and Tantau~\cite{DBLP:conf/focs/ElberfeldJT10}.
Also, recently Bodlaender, Jaffke, and Telle gave additional structural insights on the technique of typical sequences used in the algorithm~\cite{bodlaender2021typical}.
While exact computation of treewidth is known to be NP-complete~\cite{arnborg1987complexity}, currently the strongest hardness result against approximation is that by Wu, Austrin, Pitassi, and Liu, assuming the small set expansion conjecture there is no polynomial-time \mbox{$c$-approximation} algorithm for treewidth for any constant $c$~\cite{DBLP:journals/jair/WuAPL14}.
The best known polynomial-time approximation algorithm has approximation ratio $\OO(\sqrt{\log k})$~\cite{DBLP:journals/siamcomp/FeigeHL08}.

\paragraph{Outline of the algorithm}
Our algorithm is based on applying incremental improvement operations to a tree decomposition.
These improvement operations are inspired by a proof of Bellenbaum and Diestel~\cite{DBLP:journals/cpc/BellenbaumD02}, in particular, by the proof of Theorem~3 therein.
This proof shows that if a tree decomposition has a bag that is not ``lean'',\footnote{We omit the definition of ``lean'' because it is ultimately not used in our algorithm.} then the tree decomposition can be improved by a specific improvement operation.
We generalize this improvement operation, and show that with our operation a tree decomposition can be improved as long as its width is larger than $2k+1$, where $k$ is the treewidth of the graph.

In particular, we say that a bag $W \subseteq V(G)$ of a tree decomposition $T$ of a graph $G$ is \emph{splittable} if the vertex set $V(G)$ of $G$ can be partitioned into $(C_1, C_2, C_3, X)$ so that $|(C_i \cap W) \cup X| < |W|$ holds for each $i \in \{1,2,3\}$ and there are no edges between $C_i$ and $C_j$ for $i \neq j$.
We show that any bag $W$ of size $|W| > 2k+2$ is splittable, where $k$ is the treewidth of $G$.
Such bag $W$ and partition $(C_1, C_2, C_3, X)$ of $V(G)$ is then used to construct an improved tree decomposition, having width at most the width of $T$ and strictly less bags of size $|W|$.

It follows that when starting with an initial tree decomposition $T$ of width $\OO(k)$, we can obtain a tree decomposition of width $2k+1$ by iteratively applying $\OO(nk)$ improvement operations on largest bags, each of which could be implemented in $2^{\OO(k)} n$ time by finding the partition $(C_1, C_2, C_3, X)$ with standard dynamic programming techniques on $T$.
This approach would result in a $2^{\OO(k)} n^2$ time algorithm.

To optimize the dependency on $n$ to be linear, we first show that the improvement operations to the tree decomposition can be implemented so that over the course of the algorithm, in total $2^{\OO(k)} n$ bag edits are made, and the bag edits of each improvement operation are limited to a connected subtree around the splittable bag $W$.
We then implement computing the partitions $(C_1, C_2, C_3, X)$ by dynamic programming over the tree decomposition $T$ that we are also editing at the same time.
When editing the tree decomposition, we recompute the dynamic programming states only for the edited bags.
Then the process of applying improvement operations to the tree decomposition is implemented by ``walking'' over the tree decomposition in a way that maintains dynamic programming tables directed towards the current node.
Whenever a splittable largest bag $W$ is encountered in this walk, the improvement operation is performed in a local manner, recomputing the tree decomposition and dynamic programming states only for the edited subtree around $W$.

Our algorithm depends on having an initial tree decomposition of width $\OO(k)$ as an input.
By the compression technique of Bodlaender~\cite{DBLP:journals/siamcomp/Bodlaender96} (Lemma~2.7 in~\cite{DBLP:journals/siamcomp/BodlaenderDDFLP16}), any approximation algorithm for treewidth that outputs a tree decomposition of width $\alpha(k)$ can be assumed to have a tree decomposition of width at most $2\alpha(k)+1$ as an input, incurring an overhead of factor $k^{\OO(1)}$ to the running time.
In particular, in our algorithm we can assume a tree decomposition of width $4k+3$ as an input.
Our algorithm does not depend on black-box use of any other results than this compression technique.

The remaining part of this paper is organized as follows.
In \Cref{sec:preli} we provide definitions and preliminary results.
In \Cref{sec:impr} we introduce the tree decomposition improvement operation and prove its main graph-theoretical properties.
In \Cref{sec:local} we modify the improvement operation to work in a ``local'' manner and bound the number of bag edits over the course of the algorithm.
In \Cref{sec:opt} we give an efficient implementation of the improvement operations with dynamic programming.
We conclude in \Cref{sec:conc}.

\section{Preliminaries}
\label{sec:preli}
\subsection{Notation}
We use the convention that a partition of a set may contain empty parts.

The vertices of a graph $G$ are denoted by $V(G)$ and edges by $E(G)$.
The neighborhood of a vertex set $X \subseteq V(G)$ is $N(X) = \{v \in V(G) \setminus X \mid u \in X, uv \in E(G)\}$.
The subgraph induced by a vertex set $X \subseteq V(G)$ is denoted by $G[X]$ and the subgraph induced by a vertex set $V(G) \setminus X$ is denoted by $G \setminus X$.
We use the convention that a connected component is a set of vertices, i.e., a connected component of $G$ is a maximal set $C \subseteq V(G)$ so that $G[C]$ is connected.
A path is a sequence $v_1, v_2, \ldots, v_p$ of distinct vertices so that $v_i v_{i+1} \in E(G)$ for each $1 \le i < p$.
Let $X$ and $Y$ be possibly overlapping vertex sets.
A vertex set $S$ separates $X$ from $Y$ (and $Y$ from $X$) if every path that intersects both $X$ and $Y$ intersects also $S$.

A tree is a connected acyclic graph.
A subtree is a connected induced subgraph of a tree.
To distinguish trees from graphs, vertices of a tree are called nodes. 
A rooted tree has one node $r$ chosen as a root.
A node $j$ of a rooted tree is a descendant of a node $i$ if $\{i\}$ separates $\{j\}$ from $\{r\}$.
Conversely, such $i$ is an ancestor of such $j$.
If $i \neq j$ the ancestor/descendant relation may be called \emph{strict} and if $i$ and $j$ are adjacent they may be called parent/child.
A rooted subtree rooted at a node $i$ of a rooted tree is the subtree induced by the descendants of the node $i$.

\subsection{Tree decompositions}
A tree decomposition of a graph $G$ is a tree $T$ each of whose nodes $i \in V(T)$ is associated with a set $B_i \subseteq V(G)$ called a bag so that
\begin{enumerate}
\item \label{tddef:vertcond} $V(G) = \bigcup_{i \in V(T)} B_i$,
\item \label{tddef:edgecond} for each $uv \in E(G)$ there is $i \in V(T)$ with $\{u, v\} \subseteq B_i$, and
\item \label{tddef:conncond} for each $v \in V(G)$ the subgraph of $T$ induced by nodes $\{i \in V(T) \mid v \in B_i\}$ is connected (and therefore is a subtree).
\end{enumerate}

We call the above three conditions the conditions \ref{tddef:vertcond}-\ref{tddef:conncond} of tree decompositions.
The condition~\ref{tddef:conncond} is also called the \emph{connectedness condition}.

The width of a tree decomposition is $\max_{i \in V(T)} |B_i|-1$ and the treewidth of a graph $G$ is the minimum width of a tree decomposition of $G$.
We will often abuse notation by talking about a bag $B_i$ instead of the node $i$.
We usually treat a tree decomposition $T$ as rooted at some selected root node, whose bag is usually denoted by $W$.
With respect to the root bag $W$, the \emph{home bag} of a vertex $x \in V(G)$ is the bag containing $x$ that is the closest to the root $W$ among all bags containing $x$.
Note that all bags containing $x$ are descendants of the home bag of $x$.

We will use the following standard lemma implicitly throughout the paper.
\begin{lemma}
Let $T_1$, $T_2$ be subtrees of a tree decomposition $T$ of a graph $G$ that are separated by a node $i$ of $T$.
The vertex sets $V_1 = \bigcup_{j \in V(T_1)} B_j$ and $V_2 = \bigcup_{j \in V(T_2)} B_j$ are separated by $B_i$ in $G$.
\end{lemma}

Let $W \subseteq V(G)$ be a vertex set in $G$.
A \emph{balanced separator} of $W$ is a vertex set $X$ such that for each connected component $C$ of $G \setminus X$ it holds that $|W \cap C| \le |W|/2$.
The existence of balanced separators with size bounded by treewidth is a classical lemma from Graph Minors II.

\begin{lemma}[\cite{DBLP:journals/jal/RobertsonS86}]
\label{lem:balsep}
Let $G$ be a graph of treewidth $k$ and $W \subseteq V(G)$ a vertex set in $G$.
There is a balanced separator $X$ of $W$ of size $|X| \le k+1$.
\end{lemma}

\subsection{Bodlaender's compression technique}
The only black-box result that our algorithm relies on is the compression technique of Bodlaender, introduced in~\cite{DBLP:journals/siamcomp/Bodlaender96} and also exploited in~\cite{DBLP:journals/siamcomp/BodlaenderDDFLP16}.
Now we briefly discuss on how we use the technique.
For more details on the technique we refer to~\cite{DBLP:journals/siamcomp/Bodlaender96} and on its application to approximation to~\cite{DBLP:journals/siamcomp/BodlaenderDDFLP16}.

\begin{proposition}[\cite{DBLP:journals/siamcomp/Bodlaender96}]
\label{pro:bodl}
There is an algorithm that, given an $n$-vertex graph $G$ and an integer $k$, in $k^{\OO(1)} n$ time either
\begin{enumerate}
\item\label{ec1} determines that the treewidth of $G$ is larger than $k$,
\item\label{ec2} returns a matching in $G$ with at least $n/\OO(k^6)$ edges, or
\item\label{ec3} returns a graph $G'$ with at most $n - n/\OO(k^6)$ vertices so that the treewidth of $G'$ is at most the treewidth of $G$ if the treewidth of $G$ at most $k$, and furthermore, any tree decomposition of $G'$ of width $\le k$ can be turned into a tree decomposition of $G$ of width $\le k$ in $k^{\OO(1)} n$ time.
\end{enumerate}
\end{proposition}

In particular, in our case \Cref{pro:bodl} is exploited in the following way (which is very similar to~\cite{DBLP:journals/siamcomp/BodlaenderDDFLP16}; we include a proof for the convenience of the reader).

\begin{lemma}
\label{lem:compression}
Suppose there is an algorithm $A$ that, given an $n$-vertex graph $G$, integer $k$, and a tree decomposition of $G$ of width at most $4k+3$, in time $f(k) n$ either outputs a tree decomposition of $G$ of width at most $2k+1$ or determines that the treewidth of $G$ is larger than $k$.
Then there is an algorithm that in time $f(k) k^{\OO(1)} n$ does the same, but without requiring a tree decomposition as an input.
\end{lemma}
\begin{proof}
We use a recursive procedure, which is always called with parameters $G_r$ and $k$, where $k$ is the original input $k$ and $G_r$ is a graph whose treewidth is guaranteed to be at most $k$ if the treewidth of the input graph $G$ is at most $k$.
Each recursive call either determines that the treewidth of $G_r$ (and thus also the treewidth of $G$) is larger than $k$, or returns a tree decomposition of $G_r$ of width at most $2k+1$.

The base case of the recursion is a trivial edgeless graph with treewidth $0$.
In the start of each recursive call we use the algorithm of \Cref{pro:bodl} with parameters $(G_r, 2k+1)$.
In case~\ref{ec1}, we can return immediately.
In case~\ref{ec3}, we call the procedure recursively with the graph $G_r'$, and in case of a positive answer construct the tree decomposition of $G_r$ of width at most $2k+1$ and return it.
In case~\ref{ec2}, when the algorithm returns a matching $M \subseteq E(G)$, we contract the edges in $M$ to obtain a graph $G_r^M$ and call the algorithm recursively on $G_r^M$.
Contracting edges does not increase treewidth, so the treewidth of $G_r^M$ is at most the treewidth of $G_r$.
Also, we can obtain a tree decomposition $T$ of $G_r$ of width at most $4k+3$ from a tree decomposition of $G_r^M$ of width at most $2k+1$ by expanding the bags according to the matching, in particular, by replacing each occurrence of a vertex $w_{uv} \in V(G_r^M)$ corresponding to a contracted edge $uv \in M$ by the vertices $u$ and $v$.
Then we use the algorithm $A$ with $T$ to either get a tree decomposition of width at most $2k+1$ or to determine that the treewidth of $G_r$ is larger than $k$.

In each recursive call the number of vertices shrinks by a factor $1 - 1/\OO(k^6)$, and therefore the total time complexity is $T(n, k) = k^{\OO(1)} n + f(k)n + T(n - n/\OO(k^6), k)$, which can be bounded by $T(n, k) = f(k) k^{\OO(1)} n$.
\end{proof}

Therefore we will focus on giving the algorithm $A$ of \Cref{lem:compression}, i.e., on the problem of improving the width of a given tree decomposition from $4k+3$ to $2k+1$ or determining that the treewidth is larger than~$k$.

\section{Improving a tree decomposition}
\label{sec:impr}
In this section we describe the tree decomposition improvement operation and prove its main graph-theoretical properties.
Many ideas of this section are inspired by the work of Bellenbaum and Diestel~\cite{DBLP:journals/cpc/BellenbaumD02}, but none of our proofs is directly from therein.

\subsection{Splittable vertex sets}
Let $G$ be a graph.
We say that a vertex set $W \subseteq V(G)$ is \emph{splittable} if $V(G)$ can be partitioned into $(C_1, C_2, C_3, X)$ such that there are no edges between $C_i$ and $C_j$ for $i \neq j$ and $|(W \cap C_i) \cup X| < |W|$ holds for all $i \in \{1,2,3\}$.
We refer to such 4-tuple as a \emph{split} of $W$.

The next lemma shows that if a tree decomposition of a graph $G$ has width larger than $2k+1$, then either a largest bag of the tree decomposition is splittable or the treewidth of $G$ is larger than~$k$.

\begin{lemma}
\label{lem:part}
Let $G$ be a graph of treewidth $\le k$.
Any vertex set $W \subseteq V(G)$ of size $|W| \ge 2k+3$ is splittable.
\end{lemma}
\begin{proof}
By \Cref{lem:balsep}, there is a balanced separator $X$ of $W$ of size $|X| \le k+1$.
Now, each of the connected components $C_i$ of $G \setminus X$ has $|C_i \cap W| \le |W|/2$, but there may be more than three of these connected components.
We claim that these connected components $C_i$ can be merged into three sets $C_1$, $C_2$, and $C_3$ that form a partition $(C_1, C_2, C_3)$ of $V(G) \setminus X$, so that for each $i \in \{1,2,3\}$ it holds that $|C_i \cap W| \le |W|/2$.

This merging can be achieved by a following process.
We maintain a partition $(C_1, \ldots, C_t)$ of $V(G) \setminus X$.
If the partition has at least four parts, let $C_1$ and $C_2$ be the two parts with the smallest values of $|C_i \cap W|$.
We replace $C_1$ and $C_2$ by their union $C_1 \cup C_2$.
In this case it must hold that $|(C_1 \cup C_2) \cap W| \le |W|/2$, because there were at least four parts, and $C_1$ and $C_2$ were the two parts with the smallest values of $|C_i \cap W|$.

We end up with a partition $(C_1, C_2, C_3)$ of $V(G) \setminus X$, so that $|C_i \cap W| \le |W|/2$ for each $i$ and there are no edges between $C_i$ and $C_j$ for $i \neq j$.
The partition $(C_1, C_2, C_3, X)$ is indeed a split of $W$ because $|(W \cap C_i) \cup X| \le |W|/2 + k+1 < |W|$.
\end{proof}

In the algorithm, the set $W$ will always be a largest bag of the tree decomposition, i.e., a bag of size $|W| = w+1$, where $w$ is the width of the tree decomposition, and we will consider the tree decomposition as rooted at $W$.

We impose additional restrictions on the splits that we consider.
Let $W$ be root bag of a rooted tree decomposition $T$.
A split $(C_1, C_2, C_3, X)$ of $W$ is a \emph{minimum split} of $W$ if the split minimizes $|X|$ among all splits of $W$, and among splits minimizing $|X|$ the split minimizes $d_T(X) = \sum_{x \in X} d_T(x)$, where $d_T(x)$ is the distance from the home bag of $x$ to the root bag $W$ in $T$.
The improvement operation will be performed by using a minimum split of $W$.

\subsection{The improvement operation}
\label{sec:split}
We describe the construction of an improved tree decomposition by using a minimum split.
This construction is illustrated with an example in \Cref{fig:improvement}.

\begin{figure}[!t]
\begin{center}
\begin{tikzpicture}[scale=1, every node/.style={draw, rectangle, minimum size=13pt, text centered,inner sep=2.5pt,scale=1}]
\node[draw=none] at (0,1.4) {\large $T$};

\node (W) at (0,0.8) {$a,b,c,d,e$};
\node[draw=none] (Ww) at (-1.2,0.8) {$W$};

\node (B) at (-1.2,-0.2) {$a,b,f,g$};
\node[draw=none] (Bb) at (-2.2,-0.2) {$B_1$};

\node (C) at (-1.2,-1.2) {$a,f,g,h$};
\node[draw=none] (Cc) at (-2.2,-1.2) {$B_2$};

\node (D) at (1.2,-0.2) {$d,e,i,j$};
\node[draw=none] (Dd) at (0.25,-0.2) {$B_3$};

\node (E) at (1.2,-1.2) {$d,i,j,k$};
\node[draw=none] (Ee) at (0.25,-1.2) {$B_4$};

\path (W) edge (B);
\path (B) edge (C);
\path (W) edge (D);
\path (D) edge (E);

\node[draw=none] at (0,-2) {\large $\downarrow$};

\node (Xc) at (0,-2.8) {$h,i$};
\node[draw=none] at (-0.6,-2.7) {\large $X$};

\node[draw=none] at (-4,-6.6) {\large $T^1$};

\node (W1) at (-4,-4) {$a,c,h,i$};
\node[draw=none] (Ww1) at (-5,-4) {$W^1$};

\node (B1) at (-4.9,-5) {$a,h$};
\node[draw=none] (Bb1) at (-5.6,-5) {$B^1_1$};

\node (C1) at (-4.9,-6) {$a,h$};
\node[draw=none] (Cc1) at (-5.6,-6) {$B^1_2$};

\node (D1) at (-3.1,-5) {$i$};
\node[draw=none] (Dd1) at (-3.6,-5) {$B^1_3$};

\node (E1) at (-3.1,-6) {$i,k$};
\node[draw=none] (Ee1) at (-3.7,-6) {$B^1_4$};

\path (W1) edge (B1);
\path (B1) edge (C1);
\path (W1) edge (D1);
\path (D1) edge (E1);

\node[draw=none] at (0,-6.6) {\large $T^2$};

\node (W2) at (0,-4) {$b,h,i$};
\node[draw=none] (Ww2) at (-0.85,-4) {$W^2$};

\node (B2) at (-0.9,-5) {$b,g,h$};
\node[draw=none] (Bb2) at (-1.7,-5) {$B^2_1$};

\node (C2) at (-0.9,-6) {$g,h$};
\node[draw=none] (Cc2) at (-1.55,-6) {$B^2_2$};

\node (D2) at (0.9,-5) {$i$};
\node[draw=none] (Dd2) at (0.4,-5) {$B^2_3$};

\node (E2) at (0.9,-6) {$i$};
\node[draw=none] (Ee2) at (0.4,-6) {$B^2_4$};

\path (W2) edge (B2);
\path (B2) edge (C2);
\path (W2) edge (D2);
\path (D2) edge (E2);

\node[draw=none] at (4,-6.6) {\large $T^3$};

\node (W3) at (4,-4) {$d,e,h,i$};
\node[draw=none] (Ww3) at (3,-4.02) {$W^3$};

\node (B3) at (3.1,-5) {$f,h$};
\node[draw=none] (Bb3) at (2.4,-5) {$B^3_1$};

\node (C3) at (3.1,-6) {$f,h$};
\node[draw=none] (Cc3) at (2.4,-6) {$B^3_2$};

\node (D3) at (4.9,-5) {$d,e,i,j$};
\node[draw=none] (Dd3) at (3.95,-5) {$B^3_3$};

\node (E3) at (4.9,-6) {$d,i,j$};
\node[draw=none] (Ee3) at (4.1,-6) {$B^3_4$};

\path (W3) edge (B3);
\path (B3) edge (C3);
\path (W3) edge (D3);
\path (D3) edge (E3);

\path (W1) edge (Xc);
\path (W2) edge (Xc);
\path (W3) edge (Xc);
\end{tikzpicture}
\end{center}
\caption{Example of the improvement operation. A tree decomposition $T$ of a graph $G$ with $V(G) = \{a,b,c,d,e,f,g,h,i,j,k\}$ with root bag $W = \{a,b,c,d,e\}$ (top). For a minimum split $(C_1, C_2, C_3, X) = (\{a,c,k\}, \{b,g\}, \{d,e,f,j\}, \{h,i\})$ of $W$, the constructed improved tree decomposition (bottom). For the bag $W$ it holds that $W^X = \{h,i\}$, for the bag $B_1$ it holds that $B^X_1 = \{h\}$, and for the other bags $B_i$ it holds that $B^X_i = \emptyset$.\label{fig:improvement}}
\end{figure}

Let $T$ be a rooted tree decomposition of a graph $G$, $W$ the root bag of $T$, and $(C_1, C_2, C_3, X)$ a minimum split of $W$.
We first give a slightly informal description of the improvement operation and then a more formal description with additional notation.

For each $i \in \{1,2,3\}$, we obtain a tree decomposition $T^i$ of the induced subgraph $G[C_i \cup X]$ by removing all other vertices than $C_i \cup X$ from each bag of $T$, and then inserting each vertex $x \in X$ to all bags on the path from the root to the home bag of $x$ (excluding the home bag, which already contains $x$).
In particular, each $T^i$ will have a root bag $W^i = (W \cap C_i) \cup X$.
Then, the improved tree decomposition is obtained by combining $T^1$, $T^2$, and $T^3$ by connecting them from their root bags $W^1$, $W^2$, and $W^3$ to a new bag with vertex set $X$.

Next we define the construction of the improved tree decomposition more formally with the help of some additional notation.
First, for each bag $B$ of the tree decomposition $T$ we define the set of vertices in $X$ inserted to $B^i$ for each $i \in \{1,2,3\}$ to be $B^X = \{x \in X \setminus B \mid \text{the home bag of }x \text{ is a descendant of } B \text{ in } T\}$.
Then, we can define the tree decomposition $T^i$.

\begin{definition}[The tree decomposition $T^i$]
Let $T$ be a rooted tree decomposition with root bag $W$ and $(C_1, C_2, C_3, X)$ a minimum split of $W$.
For each $i \in \{1,2,3\}$, the rooted tree decomposition $T^i$ is obtained by replacing each bag $B$ of $T$ by a bag $B^i = (B \cap (C_i \cup X)) \cup B^X$.
\end{definition}

In other words, for each bag $B$ of $T$, each tree decomposition $T^i$ contains a bag $B^i$ corresponding to $B$ that is obtained by first removing all vertices not in $(C_i \cup X)$ and then inserting the set $B^X$.
The insertions of $B^X$ can be seen as first adding $X$ to the root bag $W^i$, and then ``fixing'' the connectedness condition by adding vertices $x \in X$ to bags in a minimal way.
In particular, they ensure that if $P^i$ is a parent bag of a bag $B^i$ in $T^i$, then $B^i \cap X \subseteq P^i \cap X$.
Therefore it follows that $T^i$ is a tree decomposition of the induced subgraph $G[C_i \cup X]$.

The \emph{improved tree decomposition} of $T$ with respect to $W$ and $(C_1, C_2, C_3, X)$ is then obtained by taking the disjoint union of $T^1$, $T^2$, and $T^3$ and connecting each of them from their roots $W^1$, $W^2$, and $W^3$ to a new bag with a vertex set $X$.

\begin{observation}
The improved tree decomposition is a tree decomposition of $G$.
\end{observation}
\begin{proof}
As argued above, for each $i \in \{1,2,3\}$, the tree decomposition $T^i$ is a tree decomposition of the graph $G[C_i \cup X]$.
As $X$ pairwise separates the vertex sets $C_i$ from each other, it follows that each edge and vertex of $G$ is in some of the induced subgraphs $G[C_i \cup X]$, and therefore as $T^i$ is a tree decomposition of $G[C_i \cup X]$, the improved tree decomposition satisfies the conditions~\ref{tddef:vertcond} and~\ref{tddef:edgecond} of tree decompositions.
The connectedness condition for vertices not in $X$ follows from the fact that each $T^i$ satisfies the connectedness condition and that each vertex not in $X$ appears in exactly one $T^i$.
For vertices in $X$, the connectedness condition is satisfied because by construction it is satisfied for each $T^i$ and it holds that $X \subseteq W^i$.
\end{proof}

\begin{figure}[!t]
\begin{center}
\begin{tikzpicture}[scale=0.7, every node/.style={draw, circle, minimum size=15pt, text centered,inner sep=0pt,scale=1}]
\draw[rounded corners=2mm, line width=0.2mm] (-1.5,2)--(1.5,2)--(1.5,-2)--(-1.5,-2)--cycle;

\draw[rounded corners=2mm, line width=0.2mm, fill=gray, fill opacity=0.2] (-1.5,2)--(1.5,2)--(1.5,1)--(-1.5,1)--cycle;

\draw[rounded corners=2mm, line width=0.2mm, fill=gray, fill opacity=0.2] (-1.5,0)--(1.5,0)--(1.5,-0.8)--(-1.5,-0.8)--cycle;

\draw[rounded corners=2mm, line width=0.2mm, fill=darkgray, fill opacity=0.3] (-1,2)--(-1,-2)--(1,-2)--(1,2)--(0.5,2)--(0.5,-1.4)--(-0.5,-1.4)--(-0.5,2)--cycle;

\node[draw=none] (BX) at (0,-1.7) {\large$B^X$};

\node[draw=none] (c1) at (-1.2,2.4) {\large$C_1$};
\node[draw=none] (c2) at (0,2.4) {\large$C_3$};
\node[draw=none] (c3) at (1.2,2.4) {\large$C_2$};
\node[draw=none] (bb) at (1.9,-0.4) {\large$B$};
\node[draw=none] (ww) at (1.9,1.5) {\large$W$};
\node[draw=none] (xx) at (0,-2.3) {\large$X$};

\node[draw=none] (ar) at (3.1,0) {\large$\rightarrow$};

\draw[rounded corners=2mm, line width=0.2mm] (4.5,2)--(7.5,2)--(7.5,-2)--(4.5,-2)--cycle;

\draw[rounded corners=2mm, line width=0.2mm, fill=gray, fill opacity=0.2] (4.5,2)--(7.5,2)--(7.5,1)--(4.5,1)--cycle;

\draw[rounded corners=2mm, line width=0.2mm, fill=gray, fill opacity=0.2] (4.5,0)--(7.5,0)--(7.5,-0.8)--(4.5,-0.8)--cycle;

\draw[rounded corners=2mm, line width=0.2mm, fill=darkgray, fill opacity=0.3] (5,2)--(5,0)--(4.5,0)--(4.5,-0.8)--(5.5,-0.8)--(5.5,2)--cycle;

\draw[rounded corners=2mm, line width=0.2mm, fill=darkgray, fill opacity=0.3] (7,2)--(7,0)--(7.5,0)--(7.5,-0.8)--(6.5,-0.8)--(6.5,2)--cycle;

\node[draw=none] (c12) at (4.8,2.4) {\large$C'_1$};
\node[draw=none] (c22) at (6,2.4) {\large$C'_3$};
\node[draw=none] (c32) at (7.2,2.4) {\large$C'_2$};
\node[draw=none] (bb2) at (7.9,-0.4) {\large$B$};
\node[draw=none] (ww2) at (7.9,1.5) {\large$W$};
\node[draw=none] (xx2) at (6,-2.3) {\large$X'$};
\end{tikzpicture}
\end{center}
\caption{Constructing a split $(C'_1, C'_2, C'_3, X')$ of $W$ from a split $(C_1, C_2, C_3, X)$ of $W$ in the proof of \Cref{lem:main}. The light gray illustrates the bags $W$ and $B$ and the dark gray the set $X$. The set $B^X$ is the part of $X$ that is below the bag $B$.\label{fig:mainlemma}}
\end{figure}

Next we prove the main lemma for arguing that that the improved tree decomposition is indeed improved.
The structure of the proof is to assume otherwise and then construct a split $(C'_1, C'_2, C'_3, X')$ that would contradict the fact that $(C_1, C_2, C_3, X)$ is a minimum split.
This argument is illustrated in \Cref{fig:mainlemma}.

\begin{lemma}
\label{lem:main}
Let $T$ be a rooted tree decomposition, $W$ the root bag of $T$, and $(C_1, C_2, C_3, X)$ a minimum split of $W$.
For each bag $B$ of $T$ and any pair of distinct $i,j \in \{1,2,3\}$ it holds that either $|B^X| < |B \cap (C_i \cup C_j)|$ or $B^X = \emptyset$.
\end{lemma}
\begin{proof}
By symmetry, we assume without loss of generality that $i=1,j=2$.
Suppose that $|B^X| \ge |B \cap (C_1 \cup C_2)|$ and $B^X$ is non-empty.
We claim that there is a split $(C'_1, C'_2, C'_3, X')$ of $W$ with $X' = (X \setminus B^X) \cup (B \cap (C_1 \cup C_2))$.
This split would contradict the minimality of the original split because $|X'| \le |X|$ and the home bags of vertices in $B^X$ are strict descendants of $B$ and thus strict descendants of the home bags of vertices in $B \cap (C_1 \cup C_2)$, implying that $d_T(y) < d_T(x)$ for all $y \in B \cap (C_1 \cup C_2)$ and $x \in B^X$.

To show that there is indeed such a split $(C'_1, C'_2, C'_3, X')$, first note that $B^X$ does not intersect $W$ because $B$ separates $B^X$ from $W$ and $B \cap B^X = \emptyset$, so $W \cap X \subseteq W \cap X'$.
Next we prove that the vertex sets $(W \cap C_1) \setminus X'$, $(W \cap C_2) \setminus X'$, and $(W \cap C_3) \setminus X'$ are in different connected components of $G \setminus X'$.
This implies that we can partition $V(G) \setminus X'$ to $(C'_1, C'_2, C'_3)$ so that $W \cap C'_i = (W \cap C_i) \setminus X'$ and there are no edges between $C'_i$ and $C'_j$ for $i \neq j$, implying that $(C'_1, C'_2, C'_3, X')$ is a split.

Suppose that there is a path from $(W \cap C_k) \setminus X'$ to $(W \cap C_l) \setminus X'$, with $k \neq l$, in $G \setminus X'$, and by symmetry assume that $k \in \{1, 2\}$.
The path must intersect $B^X$ before intersecting other vertices of $V(G) \setminus C_k$ because $X' \cup B^X \supseteq X$ separates $C_k$ from $V(G) \setminus C_k$.
Therefore we have a path from $W \cap C_k$ to $B^X$ that is contained in $(C_k \cup B^X) \setminus X'$.
This path must have a vertex in $B$ because $B$ separates $W$ from $B^X$.
However $B \cap (C_k \cup B^X) = B \cap C_k \subseteq X'$, so this path cannot have a vertex in $B$.
\end{proof}

Because $B^i = (B \setminus (C_j \cup C_k)) \cup B^X$ where $i,j,k$ is a permutation of $1,2,3$, \Cref{lem:main} implies that $|B^i| \le |B|$ for all $B^i$, and that $|B^i| < |B|$ if $B^X$ is non-empty.
This shows that the width of the improved tree decomposition is at most the width of $T$.
Moreover, the only case when $|B^i| = |B|$ can hold is when $B \subseteq C_i \cup X$, in which case it holds that $B^i = B$ and $B^j = B \cap X$ for $j \neq i$.
Together with the fact that $|W^i| < |W|$ and $|X| < |W|$ by the definition of a split, this implies that the number of bags of size $|W|$ in the improved tree decomposition is smaller than the number of bags of size $|W|$ in $T$ if $W$ is a largest bag of $T$.

\section{Amortized local improvement}
\label{sec:local}
A direct implementation of the improvement operation of the previous section would have time complexity $\Omega(n)$, which would result in $\Omega(n^2)$ time complexity over $n$ improvements.
In this section we introduce the \emph{pruned improvement operation} that is a slightly changed version of the improvement operation of the previous section.
We show that the pruned improvement operation can be implemented so that the number of bags edited over the course of the algorithm is bounded by $2^{\OO(k)} n$, and moreover that in each pruned improvement operation the bags edited form a subtree containing the root bag.

The main idea behind the pruned improvement operation is to exploit the fact that, as was discussed in the end of the previous section, $|B^i| = |B|$ can hold only in the case when $B \subseteq C_i \cup X$.
In this case, the whole subtree rooted at $B$ will be handled in constant time by directly copying it.
In the other case, when $|B^i| < |B|$ for all $i \in \{1,2,3\}$, the work will be charged from a potential function that is initially bounded by $2^{\OO(k)} n$.

\subsection{Pruned improvement operation}
We define the pruned improvement operation which will be used instead of the improvement operation of \Cref{sec:impr}.
The pruned improvement operation is illustrated with an example in \Cref{fig:improvementpruned}.

Let $T$ be a rooted tree decomposition with root bag $W$ and $(C_1, C_2, C_3, X)$ a minimum split of $W$.
We say that a bag $B$ of $T$ is \emph{editable} if $B$ intersects at least two of $C_1$, $C_2$, $C_3$ and every ancestor of $B$ is editable.
Note that by the definition of a split the root bag $W$ always intersects at least two of $C_1$, $C_2$, $C_3$, and therefore is always editable.

Observe that a bag $B$ that is not editable has a unique highest ancestor bag $A$ (which may be $B$ itself) for which it holds that $A \subseteq C_i \cup X$ for some $i \in \{1,2,3\}$.
In this case we say that $B$ is \emph{covered} by $C_i$ (or just covered without specifying $C_i$).
When $A \subseteq X$, we define that $B$ is covered by $C_1$, but not by $C_2$ or $C_3$, implying that every bag that is not editable is covered by exactly one $C_i$.
Observe that by definition, if $B$ is covered by $C_i$ then also all of its descendants are covered by $C_i$.
In particular, $T$ can be partitioned into a subtree of editable bags containing the root, and multiple rooted subtrees, each of which has a root bag $B \subseteq C_i \cup X$ for some $i \in \{1,2,3\}$ and whose all bags are covered by $C_i$.

We also make the following observation.

\begin{observation}
\label{obs:covbx}
If a bag $B$ is covered, then $B^X = \emptyset$.
\end{observation}
\begin{proof}
The bag $B$ has an ancestor bag $A$ for which it holds that $A \subseteq C_i \cup X$ for some $i \in \{1,2,3\}$.
Now, as $|A \cap (C_j \cup C_k)| = 0$, where $i,j,k$ is a permutation of $1,2,3$, by \Cref{lem:main} it holds that $A^X = \emptyset$.
By the definition of $B^X$, it holds that $B^X \subseteq A^X$ whenever $A$ is an ancestor of $B$.
\end{proof}

Next we define the tree decomposition $T^i$ in the pruned improvement operation.

\begin{definition}[Pruned $T^i$]
\label{def:prunedti}
Let $T$ be a rooted tree decomposition with root bag $W$ and let $(C_1, C_2, C_3, X)$ be a minimum split of $W$.
For each $i \in \{1,2,3\}$, the pruned $T^i$ is obtained by replacing each bag $B$ of $T$ by
\begin{enumerate}
\item a bag $B^i = (B \cap (C_i \cup X)) \cup B^X$ if $B$ is editable,
\item a bag $B^i = B$ if $B$ is covered by $C_i$,
\item nothing if $B$ is covered by $C_j$ for $j \neq i$.
\end{enumerate}
\end{definition}

\begin{figure}[!t]
\begin{center}
\begin{tikzpicture}[scale=1, every node/.style={draw, rectangle, minimum size=13pt, text centered,inner sep=2.5pt,scale=1}]
\node[draw=none] at (0,1.4) {\large $T$};

\node (W) at (0,0.8) {$a,b,c,d,e$};
\node[draw=none] (Ww) at (-1.2,0.8) {$W$};

\node (B) at (-1.2,-0.2) {$a,b,f,g$};
\node[draw=none] (Bb) at (-2.2,-0.2) {$B_1$};

\node (C) at (-1.2,-1.2) {$a,f,g,h$};
\node[draw=none] (Cc) at (-2.2,-1.2) {$B_2$};

\node (D) at (1.2,-0.2) {$d,e,i,j$};
\node[draw=none] (Dd) at (0.25,-0.2) {$B_3$};

\node (E) at (1.2,-1.2) {$d,i,j,k$};
\node[draw=none] (Ee) at (0.25,-1.2) {$B_4$};

\path (W) edge (B);
\path (B) edge (C);
\path (W) edge (D);
\path (D) edge (E);

\node[draw=none] at (0,-2) {\large $\downarrow$};

\node (Xc) at (0,-2.8) {$h,i$};
\node[draw=none] at (-0.6,-2.7) {\large $X$};

\node[draw=none] at (-4,-6.6) {\large $T^1$};

\node (W1) at (-4,-4) {$a,c,h,i$};
\node[draw=none] (Ww1) at (-5,-4) {$W^1$};

\node (B1) at (-4.9,-5) {$a,h$};
\node[draw=none] (Bb1) at (-5.6,-5) {$B^1_1$};

\node (C1) at (-4.9,-6) {$a,h$};
\node[draw=none] (Cc1) at (-5.6,-6) {$B^1_2$};

\path (W1) edge (B1);
\path (B1) edge (C1);

\node[draw=none] at (0,-6.6) {\large $T^2$};

\node (W2) at (0,-4) {$b,h,i$};
\node[draw=none] (Ww2) at (-0.81,-4) {$W^2$};

\node (B2) at (-0.9,-5) {$b,g,h$};
\node[draw=none] (Bb2) at (-1.72,-5) {$B^2_1$};

\node (C2) at (-0.9,-6) {$g,h$};
\node[draw=none] (Cc2) at (-1.55,-6) {$B^2_2$};

\path (W2) edge (B2);
\path (B2) edge (C2);

\node[draw=none] at (4,-6.6) {\large $T^3$};

\node (W3) at (4,-4) {$d,e,h,i$};
\node[draw=none] (Ww3) at (2.98,-4.02) {$W^3$};

\node (B3) at (3.1,-5) {$f,h$};
\node[draw=none] (Bb3) at (2.4,-5) {$B^3_1$};

\node (C3) at (3.1,-6) {$f,h$};
\node[draw=none] (Cc3) at (2.4,-6) {$B^3_2$};

\node (D3) at (4.9,-5) {$d,e,i,j$};
\node[draw=none] (Dd3) at (3.95,-5) {$B^3_3$};

\node (E3) at (4.9,-6) {$d,i,j,k$};
\node[draw=none] (Ee3) at (3.95,-6) {$B^3_4$};

\path (W3) edge (B3);
\path (B3) edge (C3);
\path (W3) edge (D3);
\path (D3) edge (E3);

\path (W1) edge (Xc);
\path (W2) edge (Xc);
\path (W3) edge (Xc);
\end{tikzpicture}
\end{center}
\caption{Example of the pruned improvement operation. A tree decomposition $T$ of a graph $G$ with $V(G) = \{a,b,c,d,e,f,g,h,i,j,k\}$ with root bag $W = \{a,b,c,d,e\}$ (top). For a minimum split $(C_1, C_2, C_3, X) = (\{a,c,k\}, \{b,g\}, \{d,e,f,j\}, \{h,i\})$ of $W$, the constructed pruned improved tree decomposition (bottom). The bags $W$, $B_1$, and $B_2$ are editable, and the bags $B_3$ and $B_4$ are covered by $C_3$. Note that even though the vertex $k$ is in $C_1$, it occurs in pruned $T^3$ instead of pruned $T^1$ because the bags containing $k$ are covered by $C_3$.\label{fig:improvementpruned}}
\end{figure}

For editable bags, the construction of pruned $T^i$ is the same as the original construction of $T^i$.
For a bag $B$ that is covered by $C_i$, a copy $B^i = B$ is created to the decomposition $T^i$, but no bags $B^j$ to $T^j$ for $j \neq i$ are created.
In particular, one may think of the construction of pruned $T^i$ as first creating the original construction for the editable bags, and then for each bag $B$ that is covered by $C_i$ and whose parent bag $P$ is editable copying the subtree rooted at $B$ to $T^i$, attaching it as a child of $P^i$.

Next we show that pruned $T^i$ can be used in the improvement operation instead of the original~$T^i$.

\begin{lemma}
Let $T$ be a rooted tree decomposition, $W$ the root bag of $T$, and $(C_1, C_2, C_3, X)$ a minimum split of $W$.
The tree decomposition constructed by connecting pruned $T^1$, $T^2$, $T^3$ from their roots $W^1$, $W^2$, $W^3$ to a new bag $X$ is a tree decomposition of $G$.
\end{lemma}
\begin{proof}
First, note that for every bag $B$ of $T$, either the same bag $B$ appears in the construction or the bags $B^i = (B \cap (C_i \cup X)) \cup B^X$ for each $i \in \{1,2,3\}$ appear in the construction.
Therefore, as every vertex and edge of $G$ is in some induced subgraph $G[C_i \cup X]$ for $i \in \{1,2,3\}$, the constructed tree decomposition satisfies the conditions~\ref{tddef:vertcond} and~\ref{tddef:edgecond} of tree decompositions.

For the connectedness condition for a vertex $v \in C_i$, there are two cases.
First, if $v$ does not appear in any editable bag, then $v$ must be completely contained in the bags of a rooted subtree covered by some $C_j$, and therefore because this subtree is directly copied to pruned $T^j$, the connectedness condition is maintained.
Second, if $v$ appears in some editable bag, $v$ will appear only in pruned $T^i$.
This is because now, if there is a covered bag $B$ containing $v$, it must hold that it is covered by $C_i$, because otherwise the highest covered ancestor $A$ of $B$ would not contain $v$ but separate $B$ from the editable bags.
Therefore for each bag $B$ of the subtree containing $v$ in $T$, there will be a bag $B^i$ containing $v$ in $T^i$, and therefore the connectedness condition is satisfied for $v$.

Finally, we argue that the connectedness condition holds for each vertex $x \in X$.
To this end, we first observe that because the root bag $W$ is editable, it holds that $X \subseteq W_i$ for every $i \in \{1,2,3\}$.
Second, we show that if $x$ occurs in a non-root bag $B^i$ of pruned $T^i$, then $x$ occurs also in the parent bag $P^i$ of $B^i$ in pruned $T^i$.
If the parent bag $P$ of $B$ is editable, we have that if $x \in B^i$, then either $x \in P$ or $x \in P^X$ and thus $x \in P^i$.
If both $B$ and its parent $P$ are covered by $C_i$, we have that if $x \in B$, then $x \in P$, because $P^X = \emptyset$ by \Cref{obs:covbx}, implying that if $x \in B^i$ then~$x \in P^i$.
\end{proof}

The pruned improvement operation will be implemented by only editing the tree decomposition for the editable bags, and directly copying the covered rooted subtrees in constant time by just changing pointers.
In \Cref{sec:opt} we will argue that with the help of appropriate data structures, the pruned improvement operation can be implemented in time $2^{\OO(k)} t$, where $t$ is the number of editable bags.
In order to do this, one remaining thing to require in the improvement operation is to maintain a maximum degree 3 of the tree decomposition.
Next we give the final definition of our improvement operation that maintains maximum degree 3 by duplicating each bag $W^i$ if necessary.

\begin{definition}[Pruned improved tree decomposition]
\label{def:prunedtd}
Let $T$ be a rooted tree decomposition of maximum degree 3 with root bag $W$ and $(C_1, C_2, C_3, X)$ a minimum split of $W$.
The pruned improved tree decomposition of $T$ with respect to $W$ and $(C_1, C_2, C_3, X)$ is constructed by first constructing pruned $T^1$, $T^2$, $T^3$, then for each $i \in \{1,2,3\}$ if $W^i$ has three children $B_1^i$, $B_2^i$, and $B_3^i$, adding a new bag $W^{i'} = W^i$ connected to $W^i$, $B_1^i$, and $B_2^i$ and removing the edges between $W^i$ and $B_1^i$, $B_2^i$, and then combining $T^1$, $T^2$, and $T^3$ by connecting each $W^i$ to a new bag containing~$X$.
\end{definition}

The construction of the pruned improved tree decomposition maintains maximum degree 3 because each pruned $T^i$ has the same maximum degree as $T$, and splitting the bag $W^i$ into $W^i$ and $W^{i'}$ ensures that the degree of $W^i$ will still be 3 after connecting it to the new bag $X$.

\subsection{Amortization}
We show that the total number of editable bags over the course of a sequence of pruned improvement operations is bounded by $2^{\OO(k)} n$.
Here we use the property that the bag $W$ will always be a largest bag of $T$, i.e., the width of $T$ is assumed to be $|W|-1$.
For the amortization, we define a potential function on a tree decomposition $T$.

\begin{definition}
Let $T$ be a tree decomposition, $w$ an integer, and $B$ a bag of $T$. The $w$-potential of $B$ is
\[
\phi_w(B) = \left\{
\begin{array}{ll}
|B| \cdot 3^{|B|}, & \text{ if } |B| \le w \text{ and }\\
3|B| \cdot 3^{|B|}, & \text{ if } |B| > w.\\
\end{array}
\right.
\]
The $w$-potential of $T$ is $\phi_w(T) = \sum_{i \in V(T)} \phi_w(B_i)$.
\end{definition}

The $w$-potential of a tree decomposition $T$ of width $k$ is bounded by $\OO(3^k k |V(T)|)$.
Next we show that a pruned improvement operation on a largest bag $W$ of size $|W|=w+1$ decreases the $w$-potential by at least the number of editable bags.

\begin{lemma}
\label{lem:local}
Let $T$ be a degree-3 rooted tree decomposition of width $w$, $W$ the root bag of $T$ with $|W| = w+1$, and $(C_1, C_2, C_3, X)$ a minimum split of $W$.
If $T'$ is the pruned improved tree decomposition of $T$ with respect to $W$ and $(C_1, C_2, C_3, X)$, and $t$ is the number of editable bags, then $\phi_w(T') \le \phi_w(T) - t$.
\end{lemma}
\begin{proof}
The tree decomposition $T'$ will have four types of bags: bags $B^i$ corresponding to covered bags $B$ of $T$, bags $B^i$ corresponding to editable non-root bags of $T$, bags $W^i$ and $W^{i'}$ corresponding to the root bag $W$ of $T$, and the bag $X$.

Let $\mathcal{E}$ be the set of editable bags of $T$, excluding the root bag $W$.
Let $\mathcal{E'}$ be the set of bags in $T'$ corresponding to the bags $\mathcal{E}$ of $T$, i.e., $\mathcal{E'} = \{B^i \mid B \in \mathcal{E} \text{ and } i \in \{1,2,3\}\}$.
Define $\phi_w(\mathcal{E}) = \sum_{B \in \mathcal{E}} \phi_w(B)$ and $\phi_w(\mathcal{E'}) = \sum_{B' \in \mathcal{E'}} \phi_w(B')$.
As the contribution of covered bags is the same for $\phi_w(T')$ and $\phi_w(T)$, we get that
\[\phi_w(T') \le \phi_w(T) - \phi_w(W) + \phi_w(X) + \sum_{i \in \{1,2,3\}} \left(\phi_w(W^i) + \phi_w(W^{i'})\right) + \phi_w(\mathcal{E'}) - \phi_w(\mathcal{E}).\]

Let us start by bounding $\phi_w(\mathcal{E'}) - \phi_w(\mathcal{E})$.
By applying \Cref{lem:main} and the fact that each editable bag $B$ intersects $C_i$ for at least two different $i \in \{1,2,3\}$, we get that for every $B \in \mathcal{E}$ it holds that
\[|B^i| = |B| - |B \cap (C_j \cup C_k)| + |B^X| < |B|,\] where $i,j,k$ is a permutation of $1,2,3$.
Therefore, by $|B| \le w+1$ we get that 
\[\sum_{i \in \{1,2,3\}} \phi_w(B^i) \le 3 (|B|-1) \cdot 3^{|B|-1} \le (|B|-1) \cdot 3^{|B|} \le \phi_w(B)-1,\]
which implies that $\phi_w(\mathcal{E}) - \phi_w(\mathcal{E'}) \ge |\mathcal{E}|$, implying that 
\[\phi_w(T') \le \phi_w(T) - \phi_w(W) + \phi_w(X) + \sum_{i \in \{1,2,3\}} \left(\phi_w(W^i) + \phi_w(W^{i'})\right) - |\mathcal{E}|.\]

For bounding the potential of the bags $W^i$, $W^{i'}$, and $X$, first we observe that the definition of a split implies $|X| \le |W^{i'}| = |W^i| < |W|$.
In particular, it holds that $|X| \le |W^{i'}| = |W^i| \le w$.
As $|W| = w+1$, it holds that $\phi_w(W) \ge 9 \cdot \phi_w(W^i)$, and therefore 
\[\phi_w(W) \ge 1 + \phi_w(X) + \sum_{i \in \{1,2,3\}} \left(\phi_w(W^i) + \phi_w(W^{i'})\right),\]
which implies that
\[\phi_w(T') \le \phi_w(T) - 1 - |\mathcal{E}|,\]
which implies the conclusion, as the number of editable bags is $1 + |\mathcal{E}|$.
\end{proof}

By \Cref{lem:local}, the total number of editable bags across all operations when improving a tree decomposition of width $w$ by using pruned improvement operations on largest bags is bounded by $\phi_w(T) = 2^{\OO(w)} n$.

\section{Implementation in linear time}
\label{sec:opt}
In this section we show that our algorithm can be implemented in $2^{\OO(k)} n$ time.
We give a data structure that allows implementing the pruned improvement operation of \Cref{sec:local} in $2^{\OO(k)} t$ time, where $t$ is the number of editable bags, and in particular allows walking over the tree decomposition to perform the operation to all largest bags in a total of $2^{\OO(k)} n$ time.

\subsection{Overview}
We treat our algorithm in the form that the input consists of a graph $G$, an integer $k$, and a degree-3 tree decomposition $T$ of $G$ of width $w$, where $2k+2 \le w \le 4k+3$.
The algorithm either outputs a tree decomposition of width at most $w-1$, or concludes that the treewidth of $G$ is larger than $k$.
It is easy to see that $\OO(k)$ applications of this algorithm gives the algorithm $A$ of \Cref{lem:compression} and therefore also the algorithm of \Cref{the:main} up to a factor of $k^{\OO(1)}$ in the time complexity.

We note that given a tree decomposition $T$ of width $w$, we can obtain a tree decomposition of maximum degree 3, width $w$, and $\OO(n)$ bags in $w^{\OO(1)} (n + |T|)$ time by standard techniques~\cite{DBLP:books/sp/Kloks94}, so we will assume that the input tree decomposition $T$ has this form.

During the algorithm we maintain a degree-3 tree decomposition $T$ and a pointer to a node $r$ of $T$.
We treat $T$ as rooted at the node $r$ and denote by $W$ the bag of $r$.
We implement a data structure that supports the following queries:

\begin{enumerate}
\item Init($T$, $r$) -- Initializes the data structure with a degree-3 tree decomposition $T$ of width $w$ and a root node $r \in V(T)$ in time $2^{\OO(w)} n$.
\item Move($s$) -- Moves the pointer from $r$ to an adjacent node $s$ in time $2^{\OO(w)}$.
\item Split() -- Returns $\bot$ if the bag $W$ of $r$ is not splittable, otherwise sets the internal state of the data structure to represent a minimum split $(C_1, C_2, C_3, X)$ of $W$ and returns $\top$. Has time complexity $2^{\OO(w)}$.
\item State() -- Assuming there has been a successful Split query after the previous Init or Edit query, returns the intersection of the bag $W$ of $r$ and the minimum split $(C_1, C_2, C_3, X)$ represented by the internal state, that is, the partition $(C_1 \cap W, C_2 \cap W, C_3 \cap W, X \cap W)$ of $W$. Has time complexity $w^{\OO(1)}$.
\item Edit($T_1$, $T_2$, $p$, $r'$) -- Replaces a subtree $T_1$ of $T$ with a new subtree $T_2$, where $r \in V(T_1)$, $r' \in V(T_2)$, and $p$ is a function from the nodes of $T \setminus T_1$ whose parents are in $T_1$ to the nodes of $T_2$, specifying how $T \setminus T_1$ will be connected to $T_2$.
The pointer $r$ will be set to $r'$.
Has time complexity $2^{\OO(w)} (|T_1| + |T_2|)$.
Assumes the new $T$ to have degree-3 and width at most $w$.
\end{enumerate}

We give a detailed description of the data structure in the next subsection.
Then, in \Cref{subsec:mainalg} we give our algorithm, using the data structure.
In \Cref{subsec:finegrained} we give a more fine-grained bound for the $2^{\OO(k)}$ factor in the time complexity of the algorithm.

\subsection{The data structure}
\label{sec:dsint}
We now describe the details of the data structure.
The data structure is essentially a dynamic programming table on the underlying tree decomposition $T$, directed towards the root node $r$.
The main idea of the Move($s$) query is that moving the root $r$ to an adjacent node $s$ changes the dynamic programming tables of only the nodes $r$ and $s$, and therefore only their tables should be recomputed.
For a split query an essential idea is that while the properties of a split depend on the intersection of the root bag $W$ with the partition $(C_1, C_2, C_3, X)$, the set $W$ does not need to be ``globally specified'' to the dynamic programming because the set $W$ will also correspond to the root node of the dynamic programming.
The state queries are implemented by tracing the solution backwards in the dynamic programming, and the edit query by removing the old subtree and computing the dynamic programming tables for the new subtree in a bottom-up manner.

The dynamic programming used will be a quite standard application of dynamic programming on tree decompositions for vertex partitioning problems.
All of the $2^{\OO(w)}$ factors in the running times of the data structure operations are of form $4^w w^{\OO(1)}$, in particular the exponential factor $4^w$ arising from the number of ways a bag of size at most $w+1$ can intersect a partition $(C_1, C_2, C_3, X)$ of $V(G)$.

\subsubsection{Stored information}
Let $i$ be a node of $T$ with a bag $B_i$, and $G[T_i]$ be the subgraph of $G$ induced by vertices in the bags of the rooted subtree of $T$ rooted at the node $i$.
For each partition $(C_1 \cap B_i, C_2 \cap B_i, C_3 \cap B_i, X \cap B_i)$ of $B_i$ and integer $0 \le h \le w$ we have a table entry $U[i][(C_1 \cap B_i, C_2 \cap B_i, C_3 \cap B_i, X \cap B_i)][h]$.
This table entry stores $\bot$ if there is no partition $(C_1, C_2, C_3, X)$ of $V(G[T_i])$ such that $|X| = h$ and there are no edges between $C_1$, $C_2$, $C_3$.
If there is such a partition, then the minimum possible integer $d_{T_i}(X)$ over all such partitions is stored, defined as $d_{T_i}(X) = \sum_{x \in X} d_{T_i}(x)$, where $d_{T_i}(x) = 0$ if $x \in B_i$ and otherwise $d_{T_i}(x)$ is the distance in $T$ between the bag $B_i$ and the closest descendant bag of $B_i$ that contains $x$.
In particular, if $B_i$ is the root bag then $d_{T_i}(X)$ is the function that should be minimized on a minimum split as a secondary measure after minimizing $|X|$.

Additionally, for each node $i$ there may be an ``internal state'' stored in order to trace the dynamic programming backwards to implement the State queries after a Split query.
The internal state is a pair $((C_1 \cap B_i, C_2 \cap B_i, C_3 \cap B_i, X \cap B_i), h)$, specifying the table entry of this node to which the minimum split fixed by the previous Split query corresponds.

We note that if $B_i$ is a leaf bag then $|V(G[T_i])| \le w+1$, and therefore all entries $U[i][\ldots][\ldots]$ can be computed directly in $2^{\OO(w)}$ time.

\subsubsection{Transitions}
Let $i$ be a node with at most three child nodes $a,b,c$, in particular $i$ corresponding to bag $B_i$ and the child nodes to bags $B_a$, $B_b$, $B_c$.
We next describe how to compute in $2^{\OO(w)}$ time the table entries $U[i][\ldots][\ldots]$ given the table entries $U[a][\ldots][\ldots]$, $U[b][\ldots][\ldots]$, and $U[c][\ldots][\ldots]$.

First, we edit the stored distances $d_{T_i}(X)$ in the entries $U[\{a,b,c\}][\ldots][\ldots]$ to correspond to distances from $B_i$.
In particular, we increment the stored distance $d_{T_i}(X)$ in each entry $U[j][(C_1 \cap B_j, C_2 \cap B_j, C_3 \cap B_j, X \cap B_j)][h] \neq \bot$ by $h - |X \cap B_j \cap B_i|$.
Then we do the transition by first decomposing it into $\OO(w)$ ``nice'' transitions of types introduce, forget, and join.
In an introduce transition we have a bag $B_i$ with a single child bag $B_j$ with $B_j \subseteq B_i$ and $|B_i \setminus B_j| = 1$, in a forget transition we have a bag $B_i$ with a single child bag $B_j$ with $B_i \subseteq B_j$ and $|B_j \setminus B_i| = 1$, and in a join transition we have a bag $B_i$ with two child bags $B_j$,$B_k$ with $B_j = B_k = B_i$.
The decomposition is done by first forgetting every vertex not in $B_i$, then introducing every vertex in $B_i$, and then joining.

The transitions follow standard ideas of dynamic programming on tree decompositions and can be done in time $2^{\OO(w)}$ as follows.
We define $U[\ldots][\ldots][h] = \bot$ for all $h < 0$ and for all $h > w$.

\paragraph{Introduce} Let $\{v\} = B_i \setminus B_j$.
For each partition $(C_1 \cap B_i, C_2 \cap B_i, C_3 \cap B_i, X \cap B_i)$ of $B_i$ and each integer $0 \le h \le w$ we set 
\begin{align*}
&U[i][(C_1 \cap B_i, C_2 \cap B_i, C_3 \cap B_i, X \cap B_i)][h] =\\
&U[j][(C_1 \cap B_i \setminus \{v\}, C_2 \cap B_i \setminus \{v\}, C_3 \cap B_i \setminus \{v\}, X \cap B_i \setminus \{v\})][h - |\{v\} \cap X|],
\end{align*}
if there are no edges between $C_1 \cap B_i$, $C_2 \cap B_i$, $C_3 \cap B_i$, and otherwise to $\bot$.

\paragraph{Forget} Let $\{v\} = B_j \setminus B_i$.
For each partition $(C_1 \cap B_i, C_2 \cap B_i, C_3 \cap B_i, X \cap B_i)$ of $B_i$ and each integer $0 \le h \le w$ we set
\begin{align*}
&U[i][(C_1 \cap B_i, C_2 \cap B_i, C_3 \cap B_i, X \cap B_i)][h] =&\\
\min \{\quad&U[j][(C_1 \cap B_i \cup \{v\}, C_2 \cap B_i, C_3 \cap B_i, X \cap B_i)][h],&\\
&U[j][(C_1 \cap B_i, C_2 \cap B_i \cup \{v\}, C_3 \cap B_i, X \cap B_i)][h],&\\
&U[j][(C_1 \cap B_i, C_2 \cap B_i, C_3 \cap B_i \cup \{v\}, X \cap B_i)][h],&\\
&U[j][(C_1 \cap B_i, C_2 \cap B_i, C_3 \cap B_i, X \cap B_i \cup \{v\})][h]\quad\},&
\end{align*}
where $\min(\bot, n) = n$ for any integer $n$.

\paragraph{Join} Let $B_j$, $B_k$ be the child bags of $B_i$. 
For each integer $0 \le h \le w$ and each partition $(C_1 \cap B_i, C_2 \cap B_i, C_3 \cap B_i, X \cap B_i)$ of $B_i$ we set 
\begin{align*}
&U[i][(C_1 \cap B_i, C_2 \cap B_i, C_3 \cap B_i, X \cap B_i)][h] =&\\
\min_{h_1 + h_2 = h + |X \cap B_i|} \Big(\quad&U[j][(C_1 \cap B_i, C_2 \cap B_i, C_3 \cap B_i, X \cap B_i)][h_1]+&\\
&U[k][(C_1 \cap B_i, C_2 \cap B_i, C_3 \cap B_i, X \cap B_i)][h_2]\quad\Big),&
\end{align*}
where $\bot + n = \bot$ and $\min(\bot, n) = n$ for any integer $n$.
Note that we do not double count $d_{T_i}(x)$ for any $x \in X$ because if $x$ is in both subtrees of $j$ and $k$ then it is also in $B_i$ and therefore has $d_{T_i}(x) = 0$.

\subsubsection{Split query}
Now the Split query on the node $r$ with bag $W$ amounts to iterating over all integers $0 \le h \le w$ and partitions $(C_1 \cap W, C_2 \cap W, C_3 \cap W, X \cap W)$ of $W$ such that $|(W \cap C_i)| + h < |W|$ for all $i$, and returning $\bot$ if the entries of all of them contain $\bot$ and otherwise returning $\top$.
In the latter case, the internal state of the node $r$ will be set to a pair $((C_1 \cap W, C_2 \cap W, C_3 \cap W, X \cap W), h)$ so that $U[r][(C_1 \cap W, C_2 \cap W, C_3 \cap W, X \cap W)][h]$ is not $\bot$, primarily minimizes $h$, and secondarily minimizes the stored integer $d_{T_i}(X)$.
In particular, so that $(C_1, C_2, C_3, X)$ is a minimum split and~$|X| = h$.

Also, the internal states of all other nodes are invalidated by e.g. incrementing a global counter.

\subsubsection{Move query}
Consider a move from a node $r$ to an adjacent node $s$.
First, if there has been a successful Split query after the previous Init or Edit query, but the child nodes of $r$ do not have valid internal states, we use the internal state $((C_1 \cap W, C_2 \cap W, C_3 \cap W, X \cap W), h)$ of $r$ to compute the corresponding internal states of its child nodes by implementing the dynamic programming transitions backwards.
In particular, we can in time $2^{\OO(w)}$ find the dynamic programming states of the children of $r$ that correspond to the split fixed by the previous successful Split query, and set the internal states of the children to correspond to these dynamic programming states.
Now the node $s$ is guaranteed to have a valid internal state before we move to it, and by induction the current node $r$ is always guaranteed to have a valid internal state.

Then, when moving the root from the node $r$ to the node $s$, the only edge whose direction towards the root changes is the edge between $r$ and $s$.
Therefore for all nodes $i$ except $r$ and $s$ the subtree $T_i$ rooted at $i$ will stay exactly the same.
Thus, we first re-compute the dynamic programming table of $r$ with a single transition and then the dynamic programming table of $s$ with a single transition, taking in total $2^{\OO(w)}$ time.
Note that all of the re-computations of the dynamic programming tables happen after computing the internal states, so the internal state of each node corresponds to the dynamic programming table directed towards the node at which the previous successful Split query was applied.

\subsubsection{Init query}
The dynamic programming tables are initialized with the already described transitions in a bottom-up manner, starting from the leaves towards the root $r$.
As the initial tree decomposition $T$ has $\OO(n)$ nodes and each transition is implemented in $2^{\OO(w)}$ time, the initialization takes $2^{\OO(w)} n$ time.

\subsubsection{State query}
With the move queries we have already guaranteed that the current node $r$ has a valid internal state corresponding to a minimum split $(C_1, C_2, C_3, X)$, if indeed there has been a successful Split query after the previous Init or Edit query.
Therefore we just return the partition $(C_1 \cap W, C_2 \cap W, C_3 \cap W, X \cap W)$ of the internal state.

\subsubsection{Edit query}
Consider an edit query that replaces a subtree $T_1$ with $T_2$, where $r \in V(T_1)$.
Because $r \in V(T_1)$, all the dynamic programming tables of nodes of $T \setminus T_1$ are already oriented towards the subtree $T_1$, and therefore the tables for the inserted nodes $T_2$ can be constructed in a bottom-up manner with $|T_2|$ transitions.
Then, with at most $|T_2|$ Move operations the root $r$ can be moved to the specified root $r'$.
Therefore, the total time complexity will be $2^{\OO(w)} |T_2| + w^{\OO(1)} |T_1|$.

\subsection{The algorithm}
\label{subsec:mainalg}
We now describe our algorithm, making use of the data structure.
The goal is to traverse the given tree decomposition $T$ of width $w$ with the Move($s$) operations, and every time when a bag $W$ of size $|W| = w+1$ is encountered, to apply the pruned improvement operation.

We start by showing that when the root pointer $r$ of the data structure is on a splittable bag $W$ of size $|W| = w+1$, the pruned improvement operation can be implemented in time $2^{\OO(w)} t$, where $t$ is the number of editable bags.

\begin{lemma}
\label{lem:impr_impl}
Let the state of the data structure be so that the root $r$ is on a bag $W$ of size $|W| = w+1$.
There is an algorithm that either in time $2^{\OO(w)}$ reports that $W$ is not splittable, or in time $2^{\OO(w)} t$ transforms $T$ into the pruned improved tree decomposition of $T$ with respect to $W$ and a minimum split $(C_1, C_2, C_3, X)$ of $W$, where $t$ is the number of editable bags.
In the latter case, the pointer $r$ of the data structure will be placed to some new bag introduced by the pruned improvement operation.
\end{lemma}
\begin{proof}
First, we use the Split query.
If it returns $\bot$ we return that $W$ is not splittable.
Otherwise, it returns that $W$ is splittable and sets the internal state of the data structure to represent a minimum split $(C_1, C_2, C_3, X)$ of $W$.

Then, as the editable bags form a subtree $T_E$ of $T$, we use the Move and State queries to find the subtree $T_E$, the bags $N(V(T_E))$ neighboring $T_E$, and for all bags $B$ in $V(T_E)$ and $N(V(T_E))$ the partitions $(C_1 \cap B, C_2 \cap B, C_3 \cap B, X \cap B)$.
Because $T$ has maximum degree 3, this can be done with $\OO(|V(T_E)|)$ Move and State queries by using them to implement a depth-first search that returns from a subtree as soon as it finds a bag $B$ for which $B \subseteq C_i \cup X$ holds.

Then, we construct the pruned improved tree decomposition for editable bags.
For this, we need to determine the sets $B^X$ for all editable bags $B$.
We observe that by the definition of $B^X$, it holds that if $B$ has children $B_1$, $B_2$, and $B_3$, then 
\[B^X = B_1^X \cup B_2^X \cup B_3^X \cup (X \cap (B_1 \cup B_2 \cup B_3)) \setminus B.\]
Therefore, by using \Cref{obs:covbx} that $B^X = \emptyset$ for covered bags $B$, the sets $B^X$ can be computed in a bottom-up manner in $|V(T_E)| w^{\OO(1)}$ time.
After computing the sets $B^X$, computing the pruned improved tree decomposition $T'_E$ for the editable bags can be done directly by definition (\Cref{def:prunedti,def:prunedtd}) in $|V(T_E)| w^{\OO(1)}$ time.

Then, we use the Edit operation to replace the subtree $T_E$ of editable bags by the constructed $T'_E$.
Here, the mapping $p$ from $N(V(T_E))$ to $V(T'_E)$ is determined as follows.
Let $B_v$ be a bag of a node $v \in N(V(T_E))$, let $i \in \{1,2,3\}$ so that $B_v$ is covered by $C_i$, and let $u$ be the parent of $v$ in $T$.
Now, $u \in V(T_E)$, and in $V(T'_E)$ there are three copies $u^1$, $u^2$, and $u^3$ corresponding to $u$.
The mapping $p$ is set so that $p(v) = u^i$.
The node $r'$ is set to be an arbitrary node in $V(T'_E)$.
This implements the construction of the pruned improved tree decomposition.

In total, we used one Split operation, $\OO(|V(T_E)|)$ Move and State operations, and one Edit operation with subtrees of size $\OO(|V(T_E)|)$, and therefore the total time complexity is $2^{\OO(w)} |V(T_E)|$ which is $2^{\OO(w)} t$.
\end{proof}

Now, by \Cref{lem:local}, the total time used in the improvement operations implemented as described in the proof of \Cref{lem:impr_impl} is bounded by $\phi_w(T) 2^{\OO(w)} = 2^{\OO(w)} n$.
What is left is to show is that between the improvement operations, we can use the Move operations to move the pointer $r$ to the next bag of size $|W| = w+1$ so that the total number of Move operations used is bounded by $\phi_w(T)$.
We do this with a depth-first-search type algorithm as we next describe.

We traverse the tree decomposition in a depth-first order with Move operations.
For simplicity, we add an extra starting node $h$ with empty bag and degree 1 to the tree decomposition and set the root pointer $r$ to $h$ initially.
Note that an empty bag cannot be editable, so the node $h$ will never be edited by the pruned improvement operation.
For all nodes there are three states -- unseen, open, and closed.
At the start the node $h$ is open and other nodes are unseen.
Let $W$ denote the bag of the node $r$.
Our algorithm traverses the tree decomposition according to the following cases:
\begin{enumerate}
\item The node $r$ is open and has an unseen neighbor node $s$: Apply Move($s$) and set the node $s$ as open.
\item The node $r$ is open and has no unseen neighbors:
\begin{enumerate}
\item It holds that $r = h$: We are done, return $T$.
\item It holds that $|W| \le w$: Set $r$ as closed and apply Move($s$) where $s$ is an open neighbor of $r$.
\item It holds that $|W| = w+1$: Apply \Cref{lem:impr_impl}. If it returns that $W$ is not splittable, then return that the treewidth of $G$ is larger than $k$. Otherwise, the new nodes inserted by the pruned improvement operation are set as unseen, and then the root $r$ is moved to a node that is open and adjacent to a newly inserted node.
\end{enumerate}
\end{enumerate}

Next we prove two key invariants for arguing the correctness and time complexity of the above described procedure, in particular that despite the improvement operations, the main properties of the procedure are similar to a standard depth-first-search.
First, we show that the open nodes form a path in the tree decomposition.

\begin{lemma}
\label{lem:dfsinva1}
The above described procedure maintains the invariant that the open nodes form a path $v_1, \ldots, v_p$, where $v_1 = h$ and $v_p = r$.
\end{lemma}
\begin{proof}
The case~1 maintains the invariant by appending one vertex to the end of the path and the case~2b by removing the last vertex of the path.
For the case 2c, we first observe that the editable subtree contains the node $r$ but not $h$ because the bag of $h$ is empty.
Therefore, because $T$ is a tree and $v_1, \ldots, v_p$ is a path from $h$ to $r$, removing the editable subtree removes a suffix $v_i, \ldots, v_p$ of the path, where $i>1$.
Then, the only node on the path $v_1, \ldots, v_{i-1}$ adjacent to the newly inserted nodes is $v_{i-1}$, which is then chosen as the new root $r$, and thus the invariant is maintained.
\end{proof}

Then, we show that the open and unseen nodes form a subtree of $T$.

\begin{lemma}
\label{lem:dfsinva2}
The above described procedure maintains the invariant that the open and unseen nodes form a subtree of $T$.
\end{lemma}
\begin{proof}
The cases that could change the set of closed nodes are 2b and 2c.
In the case 2b, by \Cref{lem:dfsinva1} the neighbor of $r$ that is on the path between $r$ and $h$ is open, and the other neighbors of $r$ are closed.
Therefore, setting $r$ as closed corresponds to removing a leaf node of the subtree of open and unseen nodes.

In the case 2c, all the neighbors of the nodes removed in the pruned improvement operation are connected to the subtree of new nodes inserted, which are all set to unseen and are connected to the path of open nodes maintained by \Cref{lem:dfsinva1}.
\end{proof}

Finally, we put everything together in the following lemma.

\begin{lemma}
\label{lem:finalg}
There is an algorithm that, given an $n$-vertex graph $G$, integer $k$, and a degree-3 tree decomposition $T$ of $G$ of width $w$, where $2k+2 \le w \le 4k+3$, in time $2^{\OO(w)} n$ either outputs a tree decomposition of $G$ of width at most $w-1$ or decides that the treewidth of $G$ is larger than $k$.
\end{lemma}
\begin{proof}
The algorithm implements the above described procedure with the data structure of \Cref{sec:dsint}.
We first prove that the algorithm is correct if it terminates, and then show that it indeed terminates in $2^{\OO(w)} n$ time.

First, the algorithm is correct when it returns that the treewidth of $G$ is larger than $k$ because in that case we have a set $W \subseteq V(G)$ of size $|W| = w+1 \ge 2k+3$ that is not splittable, and by \Cref{lem:part} this implies that the treewidth of $G$ is larger than $k$.
Second, consider the case that the algorithm terminates in the case 2a.
In this case, by \Cref{lem:dfsinva1}, the only open node is the node $r=h$, and as all neighbors of $r$ are closed, \Cref{lem:dfsinva2} guarantees that all nodes except $r$ are closed.
As a node can get closed only in case 2b, in which case the bag of the node is guaranteed to have size at most $w$, this implies that all bags in this case must have size at most $w$, and therefore the treewidth of $T$ must be at most $w-1$.
Therefore the algorithm is correct if it terminates.

By \Cref{lem:local}, the total number of editable bags over the course of the algorithm is at most $\phi_w(T) = 2^{\OO(w)} n$, and therefore the total number of bags created by pruned improvement operations is also at most $2^{\OO(w)} n$.
It also implies that the total time spent in cases~2c of the procedure is bounded by $2^{\OO(w)} \phi_w(T) = 2^{\OO(w)} n$.

For bounding the Move operations applied in cases 1 and 2b of the procedure, observe that both of them advance the state of a node either from unseen to open, or from open to closed.
Therefore, the number of Move operations in these cases is bounded by two times the total number of nodes over the course of the algorithm, which is bounded by $\OO(n+\phi_w(T))$.
This gives a total time complexity of $2^{\OO(w)} n$ for the algorithm.
\end{proof}

Now, \Cref{lem:finalg} together with \Cref{lem:compression} gives \Cref{the:main}.

\subsection{Analysis of the $2^{\OO(k)}$ factor}
\label{subsec:finegrained}
We briefly give an upper bound for the $2^{\OO(k)}$ factor in the time complexity of our algorithm, in order to support our claim that this factor in our algorithm is significantly smaller than in the algorithms of~\cite{DBLP:journals/siamcomp/BodlaenderDDFLP16}.

First, we show that if the width $w$ of the given tree decomposition is at least $3k+3$, then we can use splits where $C_3 = \emptyset$.

\begin{lemma}
Let $G$ be a graph of treewidth $\le k$.
Any vertex set $W \subseteq V(G)$ of size $|W| \ge 3k+4$ has a split of form $(C_1, C_2, \emptyset, X)$.
\end{lemma}
\begin{proof}
Again, as in \Cref{lem:part}, let $X$ be a balanced separator of $W$ of size $|X| \le k+1$, and let us combine the two components $C_i$ of $G \setminus X$ with the smallest sizes of $C_i \cap W$ until we obtain a partition $(C_1, C_2, X)$ of $V(G)$.
By considering the cases of whether there is a component $C_i$ with $|W \cap C_i| \ge |W|/3$ or not, we notice that we will end up with $|W \cap C_i| \le 2|W|/3$ for both $i \in \{1,2\}$.
Therefore $(C_1, C_2, \emptyset, X)$ is a split of $W$ because $|(W \cap C_i) \cup X| \le 2|W|/3 + k+1 < |W|$.
\end{proof}

Now, if the width $w$ of the input tree decomposition is $w \ge 3k+3$, we apply a version of the algorithm that only considers 2-way splits, i.e., fixes $C_3 = \emptyset$.
Note that this also changes the definition of a minimum split to only minimize over splits with $C_3 = \emptyset$, but the proof of \Cref{lem:main} still goes through identically, in particular noting that if $C_i = \emptyset$ in the original split, then the $C'_i$ constructed for the contradiction argument will also be empty.

Then, we note that the exponential factors $2^{\OO(w)}$ in the time complexity of the data structure come from the number of ways a partition $(C_1, C_2, C_3, X)$ of $V(G)$ can intersect a bag $B$ of size at most $w+1$.
This is $\OO(4^w)$, and when $C_3 = \emptyset$ this is $\OO(3^w)$.
Therefore, when $w \ge 3k+3$, the time complexity of the algorithm is bounded by $(n + \phi_w(T)) 3^w w^{\OO(1)}$, and when $w < 3k+3$ the time complexity is bounded by $(n + \phi_w(T)) 4^w w^{\OO(1)}$.

We note that also the factors $3^{|B_i|}$ in the definition of the potential function can be replaced by factors $2^{|B_i|}$ in the case when $C_3 = \emptyset$.
Therefore, for the case when $w \ge 3k+3$, the total time complexity is $2^w 3^w w^{\OO(1)} n$, which by $w \le 4k+3$ is $\OO(2^{10.4k} n)$.
When $w < 3k+3$, the total time complexity is $3^w 4^w w^{\OO(1)} n$, which by $w < 3k+3$ is $\OO(2^{10.8k} n)$.
Here the factors polynomial in $k$ are dominated by rounding up the exponential dependency on $k$.
Therefore, the total time complexity of the algorithm is $\OO(2^{10.8k} n)$.

\section{Conclusion}
\label{sec:conc}
We gave a $2^{\OO(k)} n$ time 2-approximation algorithm for treewidth.
This is the first 2-approximation algorithm for treewidth that is faster than the known exact algorithms, and improves the best approximation ratio achieved in time $2^{\OO(k)} n$ from 5 to 2~\cite{DBLP:journals/siamcomp/BodlaenderDDFLP16}.

Our algorithm improves upon the algorithm of Bodlaender et al.~\cite{DBLP:journals/siamcomp/BodlaenderDDFLP16} also in the running time dependency on $k$ hidden in the $2^{\OO(k)}$ notation.
Bodlaender et al. do not include an analysis of this factor in their work, nor attempt to optimize this factor in any way, but we note that their algorithm makes use of dynamic programming with time complexity $\Omega(9^w)$ on a tree decomposition of width $w$, where an upper bound for $w$ is $30k$, yielding a rough estimate of $2^{95k}$ for the dependency on $k$.
While our algorithm constitutes progress in improving the dependency on $k$, the problem of finding a constant-factor treewidth approximation algorithm with running time $c^k n$, where the constant $c$ is small, remains open.
Nevertheless, we believe that despite somewhat impractical worst-case bounds, our techniques could be well applicable for practical implementations, and in fact, the MSVS heuristic proposed in~\cite{kosterthesis,DBLP:journals/networks/KosterHK02} already resembles our algorithm on some aspects.

After the conference version of this article~\cite{Korhonen21} was published, the techniques introduced in this work have been further developed in several subsequent works.
Fomin~and~Korhonen~\cite{DBLP:conf/stoc/FominK22} extended the techniques to the setting of branchwidth of symmetric submodular functions and obtained a $f(k) n^2$ time 2-approximation algorithm for rankwidth.
Korhonen~and~Lokshtanov~\cite{DBLP:journals/corr/abs-2211-07154} introduced a generalization of the improvement operation and used it, along with other techniques, to obtain a $2^{\OO(k^2)} n^4$ time exact algorithm for treewidth.
Very recently, Korhonen, Majewski, Nadara, Pilipczuk, and Soko\l{}owski~\cite{DBLP:journals/corr/abs-2304-01744} gave a dynamic algorithm for maintaining tree decompositions of small width, whose central ingredient is a refinement operation that further builds on the improvement operation of this article and its generalization by Korhonen~and~Lokshtanov~\cite{DBLP:journals/corr/abs-2211-07154}.

\section*{Acknowledgements}
I thank Otte Hein\"avaara, Mikko Koivisto, Hans Bodlaender, and Fedor V. Fomin for helpful comments.

\bibliographystyle{alpha}
\bibliography{paper}

\end{document}